\documentclass[11pt]{article}
\usepackage{dsfont}
\usepackage{environ}
\input{definitions.sty}

\newcommand{\sAtt}{{\mathsf{ Att}}}
\newcommand{\sRep}{{\mathsf{ Rep}}}

\newcommand{\Z}{{\mathbb{Z}}}
\newcommand{\R}{{\mathbb{R}}}
\newcommand{\Rp}{{\mathbb{R}^+}}
\newcommand{\Rpm}{{\mathbb{R}_+^m}}
\newcommand{\Rpn}{{\mathbb{R}_+^n}}

\newcommand{\Inv}{\text{Inv}}
\newcommand{\Int}{\text{int}}
\newcommand{\Con}{\text{Con}}

\newcommand{\img}{\text{im }}

\def\<{\langle}
\def\>{\rangle}

\begin{document}

\title{Nash, Conley, and Computation:\\ Impossibility and Incompleteness in Game Dynamics}

\author{
	    \textbf{Jason Milionis} \\
		\small Columbia University\\
		\small \tt{jm@cs.columbia.edu}
		\and
		\textbf{Christos Papadimitriou} \\
		\small Columbia University\\
		\small \tt{christos@columbia.edu}\\
		\and
		\textbf{Georgios Piliouras} \\
		\small SUTD\\
		\small \tt{georgios@sutd.edu.sg}
		\and
		\textbf{Kelly Spendlove}\\
		\small University of Oxford\\
		\small \tt{spendlove@maths.ox.ac.uk}
}
\date{}
\maketitle

\thispagestyle{empty}

\begin{abstract}
Under what conditions do the behaviors of players, who play a game repeatedly, converge to a Nash equilibrium? If one assumes that the players' behavior is a discrete-time or continuous-time rule whereby the current mixed strategy profile is mapped to the next, this becomes a problem in the theory of dynamical systems.
We apply this theory, and in particular the concepts of chain recurrence, attractors, and Conley index, to prove a general impossibility result: there exist games for which any dynamics is bound to have starting points that do not end up at a Nash equilibrium.
We also prove a stronger result for $\epsilon$-approximate Nash equilibria: there are games such that no game dynamics can converge (in an appropriate sense) to $\epsilon$-Nash equilibria, and in fact the set of such games has positive measure. Further numerical results demonstrate that this holds for any $\epsilon$ between zero and $0.09$.
Our results establish that, although the notions of Nash equilibria (and its computation-inspired approximations) are universally applicable in all games, they are also \textit{fundamentally incomplete as predictors} of long term behavior, \textit{regardless of the choice of  dynamics}.
\end{abstract}

\section{Introduction}
The Nash equilibrium, defined and shown universal by John F.~Nash in 1950 \citep{Nash1950}, is paramount in game theory, routinely considered as the default solution concept --- the ``meaning of the game.'' 
Over the years --- and especially in the past two decades during which game theory 
has come under intense computational scrutiny --- the Nash equilibrium has been noted to
suffer from a number of disadvantages of a computational nature.
There are no efficient algorithms for computing the Nash equilibrium of a game, 
and in fact the problem has been shown to be intractable \citep{DGP,CDT,EY}.
Also, there are typically many Nash equilibria in a game, and the selection problem
leads to conceptual complications and further intractability, see e.g.~\citep{GP}.
A common defense of the Nash equilibrium is the informal argument that ``the players will eventually get there.''
However, no learning behavior has been shown to converge to the Nash equilibrium, and
many {\em game dynamics} of various sorts proposed in the past have typically been shown to fail to reach a Nash equilibrium for some games (see the section on related work for further discussion).  

Given a game, the deterministic way players move from one mixed strategy profile to the next is defined in the theory of dynamical systems as a continuous function $\varphi$ assigning to  each point $x$ in the strategy space and each time $t>0$ another point $\varphi(t,x)$: the point where the players will be after time $t$; the curve $\varphi(t,x)$ parameterized by $t$ is called the {\em trajectory} of $x$. Obviously, this function must satisfy $\varphi(t',\varphi(t,x))=\varphi(t+t', x)$.  This {\em continuous time dynamics} framework avails a rich analytical arsenal, which we exploit in this paper.\footnote{  In game theory, the dynamics of player behavior (see e.g.~the next section) are often described in terms of discrete time. The concepts we invoke (including the Conley index) have mirror images in discrete time dynamics and our results hold there as well; see Remark \ref{remark:discrete} for more detail.}

A very well known and natural dynamical system of this sort are the {\em replicator dynamics} \citep{RD_first_paper}, in which the direction of motion is the best response vector of the players, while the discrete time analogue are the multiplicative weights update dynamics \citep{MWU_WMA_Littlestone,MWU_Arora} --- but of course the possibilities are boundless.  
We should note immediately that, despite its apparent sweeping generality, this framework does come with certain restrictions: the dynamics thus defined are {\em deterministic and memoryless.} Stochastic techniques, or dynamics in which the direction of motion is computed based on the history of play, are excluded. {\color{black} Of course deterministic and memoryless algorithms suffice for a non-constructive (i.e. non-convergent) proof  of Nash equilibria in general games via Brouwer's  fixed-point theorem~\citep{nash1951non}.}

At this point it seems natural to ask:  {\em are there general dynamics of this sort, which are guaranteed to converge to a Nash equilibrium in all games?}  Our main result is an impossibility theorem:

\medskip\noindent{\bf Main Result (informal)}: {\em There are games in which any continuous, or discrete time, dynamics fail to converge to a Nash equilibrium.}

\smallskip\noindent
That is to say, we exhibit games in which any dynamics must possess long term behavior which does not coincide with the set of Nash equilibria --- or even approximate Nash equilibria.  Thus the concept of Nash equilibria is
insufficient for a description of the global dynamics: for some games, any dynamics  must have asymptotic behaviors that are {\em not} Nash equilibria.  Hence  the Nash equilibrium concept is plagued by a form of {\em incompleteness:} it is incapable for capturing the full set of asymptotic behaviors of the players.

What does it mean for the dynamics to converge to the game's Nash equilibria?
In past work on the subject \citep{DemichelisRitzberger, DeMichelisGermano1}, the requirements that have been posed are the set of fixed points of the dynamics are precisely the set of Nash equilibria (or, that the remaining fixed points may be perturbed away).  Unfortunately, such requirements are insufficient for describing global dynamics, e.g. they allow {\em cycling}, a behavior that obviously means that not all trajectories converge to the Nash equilibria.

\paragraph{What is the appropriate notion of global convergence?}
Given a game and a dynamics, where does the dynamics ``converge''?
It turns out that this is a deep question, and it took close to a century to pin down.
In the benign case of two dimensions, the answer is almost simple:  no matter where you start, the dynamics will effectively either converge to a point, or it will cycle.  This is the Poincar\'e-Bendixson Theorem from the beginning of the 20th century \citep{SmaleChaosBook}, stating that the asymptotic behavior of any dynamics is either a fixed point or a cycle (or a slightly more sophisticated configuration combining the two).  The intuitive reason is that in two dimensions {\em trajectories cannot cross,} and this guarantees some measure of good behavior.  This is immediately lost in three or more dimensions (that is, in games other than two-by-two), since dynamics in high dimensions can be chaotic, and hence convergence becomes meaningless. Topologists strived for decades to devise a conception of a ``cycle'' that would restore the simplicity of two dimensions, and in the late 1970s they succeeded!

The definition is simple (and has much computer science appeal). Ordinarily a point is called {\em \color{black} periodic} with respect to specific dynamics if it will return to itself: it is either a fixed point, or it lies on a cycle. {\color{black} If we slightly expand our definition to allow points that get arbitrarily close to where they started infinitely often then we get the notion of {\emph recurrent points}.} Now let us generalize this as follows:  a point $x_0$ is {\em chain recurrent} if for every $\epsilon>0$ there is an integer $n$ and a cyclic sequence of points $x_0, x_1,\ldots,x_n =x_0$ such that for each $i<n$ the dynamical system will bring point $x_i$ {\em inside the $\epsilon$-ball around $x_{i+1}$.}  That is, the system, started at $x_0$, will cycle after $n-1$ segments of dynamical system trajectory, interleaved with $<\epsilon$ jumps. Intuitively, it is as if an adversary manipulates the round-off errors of our computation to convince us that we are on a cycle!   

Denote by $CR(\varphi)$ the set of all chain recurrent points of the system. The main result of this theory is a decomposition theorem due to Conley, called the Fundamental Theorem of Dynamical Systems, which states that the dynamics decompose into the chain recurrent set and a gradient-like part \citep{conley1978isolated}.  Informally, the dynamical system will eventually converge to $CR(\varphi)$.  

Now that we know what convergence means, the scope of our main result becomes more clear: there is a game $g$ for which given any dynamical system $\varphi$,  $CR(\varphi)$ is \emph{not} $NE(g)$, the set of Nash equilibria of $g$.
That is, some initial conditions will necessarily cause the players to cycle, or converge to something that is not Nash, or abandon a Nash equilibrium. This is indeed what we prove.

Very inspirational for our main result was the work of~\citep{sorin_benaim}, wherein they make in passing a statement, without a complete proof, suggesting our main result: that there are games $g$ such that for all dynamics $\varphi$ (in fact, a more general class of multi-valued dynamics)  $NE(g)\neq CR(\varphi)$. The argument sketched in~\citep{sorin_benaim} ultimately rests on the development of a fixed point index for components of Nash equilibria, and its comparison to the Euler characteristic of the set of Nash. However, the argument is incomplete (indeed, the reader is directed to two papers, one of which is a preprint that seems to have not been published).   

Instead, in our proof of our main theorem we leverage an existing index theory more closely aligned to attractors: that of the {\em Conley index.}  Conley index theory \citep{conley1978isolated} provides a very general setting in which to work, requiring minimal technical assumptions on the space and dynamics. We first establish a general principle for dynamical systems stated in terms of the Conley index: if the Conley index of an attractor and that of the entire space are not isomorphic, then there is a non-empty dual repeller, and hence some trajectories are trapped away from the attractor. 

The proof of our main result applies this principle to a classical degenerate two-player game $g$ with three strategies for each player --- a variant of rock-paper-scissors due to \citep{kohlberg1986strategic}. The Nash equilibria of $g$ form a continuum, namely a six-sided closed walk in the 4-dimensional space.  We then consider an arbitrary dynamics on $g$ assumed to converge to the Nash equilibria, and show that the Conley index of the Nash equilibria attractor is not isomorphic to that of the whole space (due to the unusual topology of the former), which implies the existence of a nonempty dual repeller.  In turn this implies that $NE(g)$ is a strict subset of $CR$.
An additional algebraic topological argument shows that the dual repeller contains a fixed point, thus $NE(g)$ is in fact a strict subset of the set of fixed points.

Two objections can be raised to this result: degenerate games are known to have measure zero --- so are there dynamics that work for almost all games?  (Interestingly, the answer is ``yes, but''.) Secondly, in view of intractability, exact equilibria may be asking too much; are there dynamics that converge to an arbitrarily good approximation of the Nash equilibrium?  There is much work that needs to be dome in pursuing these research directions, but here we show two results.

First, we establish that, in some sense, degeneracy is {\em required} for the impossibility result: we give an algorithm which, given any nondegenerate game, specifies somewhat trivial dynamics whose $CR$ is precisely the set of Nash equilibria of the game.\footnote{Assuming that PPAD $\neq$ NP, this construction provides a counterexample to a conjecture by \citet{DynMeaningOfTheGame} (last paragraph of Section 5), where impossibility was conjectured unless P = NP.} The downside of this positive result is that the algorithm requires exponential time unless P = PPAD, and {\em we conjecture that such intractability is inherent.}  In other words, we suspect that, in non-degenerate games, it is {\em complexity theory,} and not topology, that provides the proof of the impossibility result. Proving this conjecture would require the development of a novel complexity-theoretic treatment of dynamics, which seems to us a very attractive research direction.

Second, we exhibit a family of games, in fact with nonzero measure (in particular, perturbations of the game used for our main result), for which any dynamics will fail to converge (in the above sense) to an $\epsilon$-approximate Nash equilibrium, for some fixed additive $\epsilon>0$ (our technique currently gives an $\epsilon$ up to about $0.09$ for utilities normalized to $[0,1]$).

\section{Related work}
The work on dynamics in games is vast, starting from the 1950s with fictitious play \citep{BrownFictitious,RobinsonFictitious}, the first of many alternative definitions of dynamics that converge to Nash equilibria in zero-sum games (or, sometimes, also in $2\times s$ games), see, e.g., \citet{KaniovskyYoung}.
There are many wonderful books about dynamics in games: \citet{Fudenberg} examine very general dynamics, often involving learning (and therefore memory) and populations of players; populations are also involved, albeit implicitly, in evolutionary game dynamics, see the book of \citet{HofbauerSigmund} for both inter- and intra-species games, the book of \citet{sandholm_population_2010} for an analysis of both deterministic and stochastic dynamics of games, the book of \citet{weibull_evolutionary_1998} for a viewpoint on evolutionary dynamics pertaining to rationality and economics and the book of \citet{hart2013simple} focusing on simple dynamics including some of the earliest impossibility results for convergence to Nash for special classes of uncoupled dynamics \citep{hart2003uncoupled}.

Regarding convergence to Nash equilibria, we have already discussed the closely related work of \citet{sorin_benaim}. The work of \citet{DemichelisRitzberger} considers \emph{Nash dynamics}, whose fixed points are precisely the set of Nash equilibria, while \citet{DeMichelisGermano1} considers {\em Nash fields}, where fixed points which are not Nash equilibria may be perturbed away; in both cases they use (fixed point) index theory to put conditions on when components of Nash equilibria are stable.  Here we point out that both Nash dynamics and Nash fields (akin to what we call Type I or Type II dynamics here) have the undesirable property of recurrence. Uncoupled game dynamics (where each player decides their next move in isolation) of several forms are considered by \citet{HartMasColell2}; some are shown to fail to converge to Nash equilibria through ``fooling arguments,'' while another {\em is shown to converge} --- in apparent contradiction to our main result. The converging dynamic is very different from the ones considered in our main theorem: it is converging to approximate equilibria (but we have results for that), and is discrete-time (our results apply to this case as well). The important departure is that the dynamics is {\em stochastic,} and such dynamics can indeed converge almost certainly to approximate equilibria.

From the perspective of optimization theory and theoretical computer science,  regret-minimizing dynamics in games has been the subject of careful investigation.
The standard approach examines their time-averaged behaviour and focuses on its convergence to coarse correlated equilibria, (see, e.g.,~\citep{roughgarden2015intrinsic,stoltz2007learning}).
This type of analysis, however, is not able to capture the evolution of day-to-day dynamics.
Indeed, in many cases, such dynamics are non-convergent in a strong formal sense even for the seminal class of zero-sum games~\citep{piliouras2014optimization,mertikopoulos2018cycles,bailey2019multi,BaileyEC18}.
Perhaps even more alarmingly strong time-average convergence guarantees may hold regardless of whether the system is divergent~\citep{BaileyEC18}, periodic~\citep{boone2019darwin}, recurrent~\citep{bailey2020finite}, or even formally chaotic~\citep{palaiopanos2017multiplicative,CFMP2019,cheung2019vortices,cheung2020chaos,bielawski2021follow,kuen2021chaos}. Recently \textit{all}  FTRL dynamics have been shown to fail to achieve (even local) asymptotic stability on \textit{any} partially mixed Nash in effectively \textit{all} normal form games despite their optimal regret guarantees \citep{flokas2020no,pmlr-v134-giannou21a}. Finally, \citet{andrade21} establish that the orbits of replicator dynamics can be \textit{arbitrarily complex}, e.g., form Lorenz-chaos limit sets, even in two agent normal form games.

The proliferation of multi-agent architectures in machine learning such as  Generative Adversarial Networks (GANs) along with the aforementioned failure of standard learning dynamics to converge to  equilibria (even in the special case of zero-sum games)
has put strong emphasis on alternative algorithms as well as the development of novel learning algorithms. In the special case of zero-sum games, several other algorithms  converge provably to Nash equilibria such as optimistic mirror descent \citep{rakhlin2013optimization,daskalakis2018training,daskalakis2018last}, the extra-gradient method (and variants thereof) \citep{korpelevich1976extragradient,gidel2019a,mertikopoulos2019optimistic}, as well as several other dynamics \citep{gidel2019negative,letcher2019differentiable,perolat2021poincare}. Such type of results raise a hopeful optimism that maybe a simple, practical algorithm exists that reliably converges to Nash equilibria in all games \textit{at least asymptotically}. 

Naturally, the difficulty of learning Nash equilibria grows significantly  when one broadens their scope to a more general class of games than merely zero-sum games~\citep{daskalakis2010learning,paperics11,galla2013complex,DynMeaningOfTheGame}. 
Numerical studies suggest that chaos is typical~\citep{sanders2018prevalence} and emerges even in low dimensional systems~\citep{sato2002chaos,palaiopanos2017multiplicative, 2017arXiv170109043P}. 
Such non-convergence results have inspired a program on the intersection of game theory and dynamical systems \citep{Entropy18,DynMeaningOfTheGame}, specifically using Conley's fundamental theorem of dynamical systems \citep{conley1978isolated}. Interestingly, it is exactly Conley index theory that can be utilized to establish a universal negative result for game dynamics, even if we relax our global attractor requirements from the set of exact Nash to approximations  thereof.  In even more general context, {\em computational} (as opposed to topological) impossibility results are known for the problem of finding price equilibria in markets \citep{PapYannPrice}. If one extends even further to  the machine learning inspired class of differential games several negative results have recently been established~\citep{letcher2021impossibility,balduzzi2020smooth,hsieh2021limits,farnia2020gans}.

\section{Preliminaries on Game Theory}

In this section we outline some fundamental game theoretic concepts, essential for the development of our paper. For more detailed background, we refer the interested reader to the books by \citet{AGT_main_book,HofbauerSigmund,weibull_evolutionary_1998}.

Consider a finite set of $K$ players, each of whom has a finite set of actions/strategies (say, for example, that player $k\in [K]$ can play any strategy $s_k$ within their strategy space $S_k$). The utility that player $k$ receives from playing strategy $s_k\in S_k$ when the other players of the game choose strategies $s_{-k} \in \prod_{l\in [K]\setminus \{k\}} S_l$ is a function $u_k : \prod_{l\in [K]} S_l \to \R$. These definitions are then multilinearly extended into their mixed/randomized equivalents, using probabilities as the weights of the strategies and utilities with respect to the strategies chosen. We call the resulting triplet $(K, \prod_{l\in [K]} S_l, \prod_{l\in [K]} u_l)$ a game in {\em normal form}.

A two-player game in normal form is said to be a {\em bimatrix} game, because the utilities in this case may equivalently be described by two matrices: $M_1\in\R^{m\times n}, M_2\in\R^{n\times m}$ for the two players respectively, when player 1 has $m$ (pure) strategies available and player 2 has $n$ (pure) strategies available. We say that such a game is an $m\times n$ bimatrix game. For mixed strategy profiles $x \in \Rpm, y \in \Rpn$ with $\sum_{i\in [m]} x_i = 1, \sum_{j\in [n]} y_j = 1$, we say that $(x, y)$ is a Nash equilibrium for the bimatrix game as above if
\[
\begin{cases}
\langle x, M_1 y \rangle \geq \langle x', M_1 y \rangle\ \text{for all } x' \in \Rpm, \sum_{i\in [m]} x_i' = 1 \\
\langle x, M_2 y \rangle \geq \langle x, M_2 y' \rangle\ \text{for all }  y' \in \Rpn, \sum_{j\in [n]} y_j' = 1.
\end{cases}
\]
(or alternatively, that $x$ is best-response of player 1 to player 2's $y$ mixed strategy, and $y$ is best-response of player 2 to player 1's $x$ mixed strategy) where $\langle \cdot, \cdot \rangle$ denotes the inner product of the respective vector space. We also say that $(x, y)$ is an $\epsilon$-approximate Nash equilibrium for the bimatrix game if
\[
\begin{cases}
\langle x, M_1 y \rangle \geq \langle x', M_1 y \rangle - \epsilon\ \text{for all } x' \in \Rpm, \sum_{i\in [m]} x_i' = 1 \\
\langle x, M_2 y \rangle \geq \langle x, M_2 y' \rangle - \epsilon\ \text{for all } y' \in \Rpn, \sum_{j\in [n]} y_j' = 1.
\end{cases}
\]

Finally, we denote by $(M_1 y)_i$ the $i$-th coordinate of the vector $M_1 y\in\R^m$.

\section{Background in Conley Theory}

In this section we review fundamental notions of dynamical systems, especially from the point of view of Conley theory.  The standard reference for chain recurrence, Conley's decomposition theorem and the Conley index is the monograph \cite{conley1978isolated}.  An approachable introduction to chain recurrence and Conley's decomposition theorem is~\cite{norton1995fundamental}, wherein the theorem is christened `The Fundamental Theorem of Dynamical Systems'.  A concise and rigorous treatment of the  theorem is given in \cite[Chapter 9]{robinson1998dynamical}.  

Attractors and repellers are of central importance in dynamical systems theory, as they form the means of filtering the global dynamics, and are dual to decomposing the behavior into recurrent and gradient-like parts.  These are particularly important topics in Conley theory, and an algebraic treatment of attractors using order theory (and developing the corresponding duality theory to Morse decompositions and chain recurrence) is given in \cite{robbin1992lyapunov,kmv,kmv2,kmv3,kaliespriestly}.  A concise review of the Conley index is found in \cite{salamon1985connected}, and overviews in \cite{mischaikow1995conley, mischaikow2002conley}.  The relationship between Conley indices is captured using connection matrix theory \cite{franzosa1986index,franzosa1988continuation,franzosa1988connection, franzosa1989connection,harker2021computational}, and our main theorem (Theorem \ref{thm:imposs}) uses elementary results in this direction.

Conley index theory has been extended to discrete time dynamics \cite{arai2002tangencies, mrozek1990leray, franks2000shift, richeson1998connection, robbin1992lyapunov, szymczak1995conley} (see applications in \cite{arai2009database,bush2012combinatorial}), as well as multi-valued dynamics~\cite{kaczynski1995conley}.

\subsection{Attractors and repellers}

We will consider a compact metric space $X$ and a dynamical system $\varphi$ on $X$
(in our application of interest, $X$ will be a product of simplices).
As we have already mentioned, the standard mathematical conception of a dynamics is a \emph{semi-flow}, a continuous map $\varphi:  \Rp \times X\to X$, where $\Rp=\{t\in \R: t\geq 0\}$, satisfying i) $\varphi(0,x)=x$ and ii) $\varphi(s+t,x)=\varphi(s,\varphi(t,x))$.  If in this definition $\Rp$ is replaced with $\R$ then $\varphi$ is called a \emph{flow}.\footnote{The distinction between semi-flows and flows is `negative time', which corresponds to invertibility of the dynamics.}  By a dynamical system, or dynamics, we mean either a flow or semi-flow.  We are interested in the asymptotic behavior of dynamics:
for a set $Y\subset X$, we define the $\omega$-limit set of $Y$ to be $\omega(Y) = \bigcap_{t>0} \text{cl}\big(\varphi([t,\infty), Y)\big)$.

One basic notion in dynamical systems is that of an {\em attractor}: intuitively, a set to which trajectories converge asymptotically.  A set $A\subset X$ is an {attractor} if there exists a neighborhood $U$ of $A$ such that $\omega(U)= A$.  This is a useful concept for our purposes: we seek a flow such that the set of Nash equilibria of a game (and collection of connecting orbits between) is an attractor.   The \emph{dual repeller} of an attractor $A$ is defined as the set
\[
R = \{x\in X: \omega (x)\cap A = \emptyset\};
\] that is, the set of points that do not converge to $A$.
The pair $(A,R)$ is called an {\em attractor-repeller pair}.  
The concept of repeller is also key for us: it is not enough that the set of Nash equilibria and connecting orbits constitute an attractor; the corresponding repeller should be empty --- otherwise some starting points will never reach the Nash equilibria.  The set of attractors of $\varphi$, denoted by $\sAtt(\varphi)$, together with binary operations:
\[
A\vee A' := A\cup A',\quad\quad A\wedge A' := \omega (A \cap A'),
\]
and partial order determined by inclusion, forms a bounded distributive lattice~\cite{kmv, kmv2, robbin1992lyapunov}.  Similarly, the set of repellers forms a bounded, distributive lattice $\sRep(\varphi)$, which is (anti)-isomorphic to $\sAtt(\varphi)$ \cite{kmv}.




\subsection{Chain recurrence}

An \emph{$(\epsilon, \tau)$ chain} from $x$ to $y$ is a finite sequence $\{(x_i,t_i)\}\subset X\times [0,\infty)$, $i=1,\ldots,n$, such that $x=x_1$, $t_i\geq \tau$, $d(\varphi(t_i,x_i),x_{i+1})\leq \epsilon$ and $d(\varphi(t_n,x_n),y)\leq \epsilon$.  If there is an $(\epsilon,\tau)$ chain from $x$ to $y$ for all $(\epsilon,\tau)$ we write $x\geq y$.  The \emph{chain recurrent set}, denoted, $CR(\varphi)$, is defined by
\[
CR(\varphi) = \{x\in X:x\geq x\}.
\]
The chain recurrent set can be partitioned into a disjoint union of its connected components, called \emph{chain recurrent components}; the chain recurrent components are partially ordered by the dynamics (inheriting the partial order of $\geq)$~\cite{conley1978isolated}.

The fundamental theorem of dynamical systems, due to Conley, states that a dynamical system can be decomposed into the (chain) recurrent set, and off of the chain recurrent set the dynamics is gradient-like (on a trajectory between two chain recurrent components). Thus there is a deep dichotomy between points in the phase space: they are either chain recurrent or gradient-like. This can be phrased in terms of a Lyapunov function as follows.

\begin{theorem}[Fundamental Theorem of Dynamical Systems]
Let $CR_i(\varphi)$ denote the chain components of $CR(\varphi)$.  There exists a continuous function $V: X\to [0,1]$ such that
\begin{enumerate}
    \item if $x\not\in CR(\varphi)$ and $t>0$ then $V(x)>V(\varphi(t,x))$;
    \item for each $i$, there exists $\sigma_i\in [0,1]$ such that $CR_i(\varphi)\subset V^{-1}(\sigma_i)$ and furthermore, the $\sigma_i$ can be chosen such that $\sigma_i\neq \sigma_j$ if $i\neq j$.
\end{enumerate}
\end{theorem}


\subsection{Conley index}

The broader focus of Conley theory is on isolating neighborhoods and isolated invariant sets.  If $\varphi$ is a semi-flow, then a set $S$ is an \emph{invariant set} if $\varphi(t,S)=S$ for all $t\in \Rp$. 
Given a set $N$ we denote the \emph{maximal invariant set} in $N$ by
\[
\Inv(N) = \bigcup \{S\subset N : \text{$S$ is an invariant set} \}.
\]
An \emph{isolating neighborhood} is a compact set $N$ such that $\Inv(N)\subset \Int (N)$.   Both an attractor and its dual repeller are examples of isolated invariant sets.

The construction of the Conley index proceeds through index pairs.  For an isolated invariant set $S$, an \emph{index pair} is a pair of compact sets $(N,L)$ with $L\subset N$ such that
\begin{enumerate}
    \item $S = \text{Inv}(\text{cl}(N\setminus L))$ and $N\setminus L$ is a neighborhood of $S$,
    \item $L$ is positively invariant in $N$, i.e., if $x\in L$ and $\varphi([0,t],x) \subset N$, then $\varphi([0,t],x)\subset L$.
    \item $L$ is an exit set for $N$, i.e., if $x\in N$ and $t_1>0$ such that $\varphi(t_1,x)\not\in N$ then there exists a $t_0\in [0,t_1]$ for which $\varphi([0,t_0],x)\subset N$ and $\varphi(t_0,x)\in L$.
\end{enumerate}

The \emph{Conley index} of an isolated invariant set $S$ is the (relative) homology group
\[
CH_\bullet(S) := H_\bullet (N,L),
\]
where $H_\bullet$ denotes singular homology with integer coefficients\footnote{See \ref{app:singular} for a primer on algebraic topology and singular homology theory.} and $(N,L)$ is an index pair for $S$.  The Conley index is independent of the particular choice of index pair and  is an algebraic topological invariant of an isolated invariant set; it is a purely topological generalization of another, integer-valued, invariant called the {\em Morse index}, see \cite{conley1978isolated, MI}. Moreover index pairs can be found so that $H_\bullet(N,L)\cong \tilde{H}_\bullet(N/L)$.   See Fig. \ref{fig:conley_indices} for elementary Conley index computations for four example isolated invariant sets.\footnote{The final three examples in Fig.~\ref{fig:conley_indices} instantiate a general result: if $S$ is a hyperbolic fixed point with an unstable manifold of dimension $n$ (i.e., the number of positive eigenvalues of the Jacobian), then $CH_i(S) = \Z$ if $i=n$, otherwise $CH_i(S) = 0$.}   In each case, an appropriate index pair $(N,L)$ is illustrated; the Conley index is computed via $CH_\bullet(S) = H_\bullet(N,L)\cong \tilde{H}_\bullet(N/L)$, the reduced homology of $N/L$.

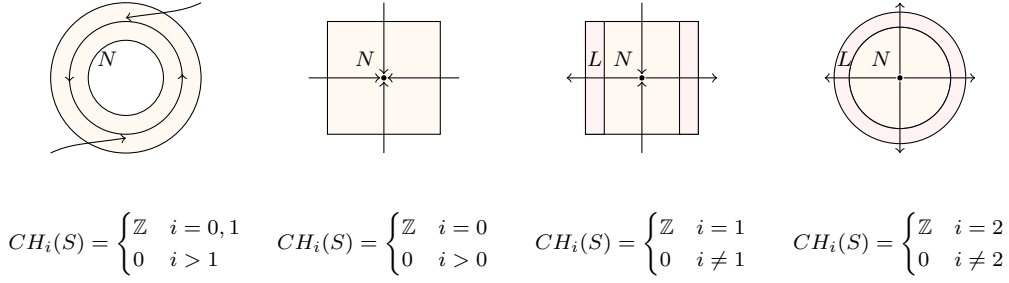
\begin{figure}[h!]
\centering
\begin{minipage}{.2\textwidth}
    \centering
    \begin{tikzpicture}[scale=.5]
    \def\x{.1}
    \def\y{2}
    \def\z{1.5}
    
    \filldraw [fill=orange!5,even odd rule] (0,0) circle[radius=2cm] circle[radius=1cm];
    \draw (0,0) circle [radius=1.5cm];
    
    \draw [->] (\z,0) -- (\z,.1);
    \draw [->] (-\z,0) -- (-\z,-.1);

    \draw (-.5,.5) node{\scriptsize $N$};

   \draw[->] (2,2) .. controls (1.5,1.75) and (1,1.75)  .. (0, 1.6);
   \draw[->] (-2,-2) .. controls (-1.5,-1.75) and (-1,-1.75)  .. (0, -1.6);

    \end{tikzpicture}
    
    {\scriptsize 
    \[ CH_i(S) = \begin{cases} 
          \Z & i = 0,1 \\
          0 & i > 1
       \end{cases}
    \]
    }
    \end{minipage}
    \begin{minipage}{.2\textwidth}
    \centering
    \begin{tikzpicture}[scale=.5]
    
    \def\x{.1}
    \def\y{2}
    \def\z{1.5}
    \filldraw[fill=orange!5] (-\z,-\z) rectangle (\z,\z);
    
    \fill (0,0) circle (2pt);
    \draw[->] (0:\y) to (0:\x);
    \draw[->] (90:\y) to (90:\x);
    \draw[->] (180:\y) to (180:\x);
    \draw[->] (270:\y) to (270:\x);

    \draw (-.5,.5) node{\scriptsize $N$};
    \end{tikzpicture}
    
    {\scriptsize 
    \[ CH_i(S) = \begin{cases} 
          \Z & i = 0 \\
          0 & i > 0
       \end{cases}
    \]
    }
    \end{minipage}
        \begin{minipage}{.2\textwidth}
    \centering
    \begin{tikzpicture}[scale=.5]
    \def\x{.1}
    \def\y{2}
    \def\z{1.5}
    \def\h{.5}
    \filldraw[fill=orange!5] (-\z,-\z) rectangle (\z,\z);
    \filldraw[fill=red!5] (-\z,-\z) rectangle (-\z+\h,\z);
    \filldraw[fill=red!5] (\z,-\z) rectangle (\z-\h,\z);

    \fill (0,0) circle (2pt);
    \draw[->] (0:\x) to (0:\y);
    \draw[->] (90:\y) to (90:\x);
    \draw[->] (180:\x) to (180:\y);
    \draw[->] (270:\y) to (270:\x);
    
    \draw (-.5,.5) node{\scriptsize $N$};
    \draw (-1.25,.5) node{\scriptsize $L$};

    \end{tikzpicture}
    
    {\scriptsize 
    \[ CH_i(S) = \begin{cases} 
          \Z & i = 1 \\
          0 & i \neq 1
       \end{cases}
    \]
    }
\end{minipage}
\begin{minipage}{.2\textwidth}
    \centering
    \begin{tikzpicture}[scale=.5]
    \def\x{.1}
    \def\y{2}
    \def\z{1.5}
    \def\h{.5}
    \filldraw[fill=orange!5] (0,0) circle (1.35cm);
    \filldraw [fill=red!5,even odd rule] (0,0) circle[radius=1.75cm] circle[radius=1.35cm];

    \fill (0,0) circle (2pt);
    \draw[->] (0:\x) to (0:\y);
    \draw[->] (90:\x) to (90:\y);
    \draw[->] (180:\x) to (180:\y);
    \draw[->] (270:\x) to (270:\y);
    
    \draw (-.5,.5) node{\scriptsize $N$};
    \draw (-1.475,.5) node{\scriptsize $L$};

    \end{tikzpicture}
    
    {\scriptsize 
    \[ CH_i(S) = \begin{cases} 
          \Z & i = 2 \\
          0 & i \neq 2
       \end{cases}
    \]
    }
\end{minipage}
    \caption{Stable periodic orbit with $N$ as the annulus and $L=\emptyset$ [left]. Asymptotically stable fixed point with $L=\emptyset$ [center left].   Hyperbolic saddle point with 1-dimensional unstable manifold [center right]. Hyperbolic unstable fixed point with 2-dimensional unstable manifold, $N$ as the ball, $L$ as the thickened boundary [right].}
    \label{fig:conley_indices}
\end{figure}

A fundamental property of the Conley index is the following.

\begin{proposition}[Wa{\.z}ewski Property]\label{prop:wazewski}
 If $S$ is an isolated invariant set and $S=\emptyset$, then $CH_\bullet(S)= 0$.  More importantly, if $CH_\bullet(S)\neq 0$ then $S\neq \emptyset$.
\end{proposition}


The Wa{\.z}ewski Property shows that the Conley index $CH_\bullet(S)$ can be used to deduce information about the structure of the associated isolated invariant set $S$.  The Wa{\.z}ewski Property is the most elementary result of this type.  More sophisicated results can be used to prove the existence of connecting orbits \cite{franzosa1989connection}, periodic orbits \cite{mccord1995zeta}, or chaos \cite{mischaikow1995chaos}, and in the case that $\varphi$ is a flow, \citet{mccord_1988} has shown Lefschetz fixed point type theorems for Conley indices.

\begin{proposition}[Fixed point theorem for Conley indices]\label{prop:fp_conley}
If $S$ is an isolated invariant set $S$ for a flow $\varphi$ and there is an $n$ such that
\[
 CH_i(S) = \begin{cases} 
      \Z & i = n \\
      0 & i \neq n,
   \end{cases}
\]
then $S$ contains a fixed point.
\end{proposition}

We prove our main result by examining the relationship between the Conley indices of an attractor-repeller pair.  These Conley indices are related by an exact sequence~\cite{conley1978isolated, franzosa1988connection,franzosa1986index}, which follows from the long exact sequence of the triple \cite{hatcher}.

\begin{proposition}[Exact sequence for attractor-repeller pair]\label{prop:les}
If $(A,R)$ is an attractor-repeller pair, then there is an exact sequence of Conley indices:
\[
\cdots \xrightarrow{j_{n+1}}CH_{n+1}(R)\xrightarrow{\partial_{n+1}} CH_n(A)\xrightarrow{i_n} CH_n(X)\xrightarrow{j_n} CH_n(R)
\xrightarrow{\partial_n} 
\cdots
\xrightarrow{j_0} CH_0(R)\to 0.
\]
\end{proposition}

\medskip\noindent
We are now in place to prove a result about Conley indices of an attractor-repeller pair, which is central to our impossibility theorem.

\begin{theorem}\label{prop:r0}
Let $(A,R)$ be an attractor-repeller pair.  If $CH_\bullet(R)=0$ then $CH_\bullet(A)$ and $CH_\bullet(X)$ are isomorphic.   
\end{theorem}
\begin{proof}
It follows from Proposition~\ref{prop:les} that there is an exact sequence relating the Conley indices of the attractor-repeller pair.  If $CH_n(R)=0$ for all $n$, this becomes
\[
\cdots \xrightarrow{j_{n+1}} 0 \xrightarrow{\partial_{n+1}} 
CH_n(A)\xrightarrow{i_n} 
CH_n(X)\xrightarrow{j_n} 0 
\xrightarrow{\partial_n} \cdots 
\]
Exactness implies that 1) $0= \img \partial_{n+1} = \ker i_n $, thus $i_n$ is injective and 2) $\img i_n = \ker j_n = CH_n(X)$, so $i_n$ is surjective.  Thus $i_n$ is an isomorphism for all $n$.
\end{proof}

\begin{corollary}\label{cor:neq}
Let $(A,R)$ be an attractor-repeller pair.  If $CH_\bullet(A)\neq CH_\bullet(X)$ then $R\neq\emptyset$. 
\end{corollary}
\begin{proof}
From Theorem~\ref{prop:r0} we have that $CH_\bullet(R)\neq 0$. Therefore $R\neq \emptyset$ by the Wa{\.z}ewski Property.
\end{proof}


\section{Game dynamics and the Conley index}
For a game $g$, we say $\varphi$ is a dynamical system {\em for} $g$ if $\varphi$ is a dynamical system on $X$, where $X$ is the space of the game $g$ (product of simplices).  Let $NE(g)$ denote the set of Nash equilibria for $g$, $Fix(\varphi)$ the set of fixed points of $\varphi$, and $CR(\varphi)$ the chain recurrent set of $\varphi$.  We highlight three compatibility conditions that $\varphi$ can take with respect to $g$.  
\begin{description}
\item[{\bf Type I}] $NE(g) \subset Fix(\varphi)$,
\item[{\bf Type II}] $NE(g) = Fix(\varphi)$,
\item[{\bf Type III}] $NE(g) = CR(\varphi)$.
\end{description}


\begin{remark}
Types I and II are only assumptions on the fixed points of the dynamics, and leave open the possibility that there are fixed points which are not Nash equilibria, or more generally, that there are recurrent dynamics that do not contain Nash equilibria (e.g., some trajectories cycle and never encounter a Nash equilibrium).  We note, however, that Type III does not assume the Nash equilibria are fixed points, i.e. Type III does not imply Type I or Type II.
\end{remark}

\begin{remark}
Note that we do not require the dynamics to be payoff increasing, uncoupled, etc.  Our only assumptions are that the dynamics are a semi-flow (continuous and deterministic).
\end{remark}




\begin{theorem}[Impossibility Theorem]\label{thm:imposs}
There exists a game $g$ that does not admit Type III dynamics.
\end{theorem}
\begin{proof}
Define $g$ to be the following bimatrix game as in~\citet{kohlberg1986strategic, sorin_benaim}:
\begin{align}\label{eqn:game}
(M_1, M_2) = \begin{pmatrix}
1,1 & 0,-1 & -1,1\\
-1,0 & 0,0 & -1,0\\
1,-1 & 0,-1 & -2,-2
\end{pmatrix}.
\end{align}

By way of contradiction, suppose that there exist Type III dynamics $\varphi$ for $g$, i.e., $NE(g)=CR(\varphi)$.  In the case of this specific $g$, the set of Nash equilibria forms a topological circle (i.e. a homeomorphic copy of $S^1$)~\cite{kohlberg1986strategic}; thus the chain recurrent set $CR(\varphi)$ consists of a single chain recurrent component. As the number of chain recurrent components is finite, there is an explicit duality to the lattice of attractors $\sAtt(\varphi)$ \cite[Theorem 6]{kmv3}. In particular, since there is just a single chain recurrent component, it is an attractor and in fact the unique (maximal) attractor.  That is, $A :=\omega(X) = NE(g)$ is the unique (maximal) attractor.


Choosing $N$ appropriately as an compact neighborhood encompassing the circle of Nash equilibria such that $A=NE(g)\subset \Int(N)$ (see Fig.~\ref{fig:index_pair}), the pair $(N,\emptyset)$ form an index pair for $A$.
Thus the Conley index of $A$ is determined via
\begin{align}\label{CH:A}
CH_i(A) = H_i(N,\emptyset) = H_i(S^1) = \begin{cases} 
      \Z & i = 0,1 \\
      0 & i > 1.
   \end{cases}
\end{align}
 On the other hand, since $X$ is a product of simplices we have that 
 \begin{align}\label{CH:X}
 CH_i(X) = H_i(X,\emptyset) = \begin{cases} 
      \Z & i = 0 \\
      0 & i > 0.
   \end{cases}
\end{align}
As $CH_\bullet(A)\neq CH_\bullet(X)$, it follows from Corollary~\ref{cor:neq} that $R\neq \emptyset$. Thus $A$ cannot be maximal and there cannot exist such a $\varphi$.
\end{proof}

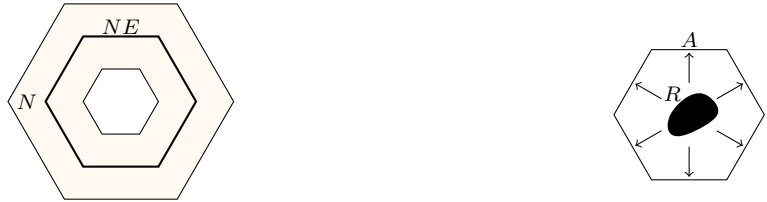
\begin{figure}[h!]
\centering
\begin{minipage}{.45\textwidth}
    \centering
    \begin{tikzpicture}[scale=.5]
    \def\R{1}
    \def\A{2}
    \def\B{3}
    \filldraw[fill=orange!5] (0:\B) \foreach \x in {60,120,...,359} {
        -- (\x:\B)
    }-- cycle ;
    \draw[thick] (0:\A) \foreach \x in {60,120,...,359} {
        -- (\x:\A)
    }-- cycle ;
    \filldraw[fill=white] (0:\R) \foreach \x in {60,120,...,359} {
        -- (\x:\R)
    }-- cycle ;
    \draw (0,2) node{\scriptsize $NE$};
    \draw (-2.5,0) node{\scriptsize $N$};

    \end{tikzpicture}
    \end{minipage}
    \begin{minipage}{.45\textwidth}
    \centering
    \begin{tikzpicture}[scale=.5]
    \def\R{1}
    \def\A{2}
    \def\B{3}
    \draw (0:\A) \foreach \x in {60,120,...,359} {
        -- (\x:\A)
    }-- cycle ;
\fill plot[smooth cycle, tension=.7] coordinates {
    (.75,0) (.5,.5) (0,.5) (-.5,0) (-.5,-.5) (0,-.5)
};
\draw [->] (30:.85) to  (30:1.65);
\draw [->] (90:.85) to  (90:1.65);
\draw [->] (150:.85) to  (150:1.65);
\draw [->] (210:.85) to  (210:1.65);
\draw [->] (-30:.85) to  (-30:1.65);
\draw [->] (-90:.85) to  (-90:1.65);

    \draw (0,2) node{\scriptsize $A$};
    \draw (-.45,.55) node{\scriptsize $R$};

    \end{tikzpicture}
\end{minipage}
    \caption{Cartoon rendering of circle of Nash equilibria ($NE(g)$ lives in 4-dimensions).  Index pair $(N,\emptyset)$  for the circle of Nash equilibria $NE$ [left]. The set of Nash equilibria cannot be a maximal attractor; the dual repeller must be nonempty [right].}
    \label{fig:index_pair}
\end{figure}




\begin{remark}\label{remark:discrete}
As mentioned earlier, an analogous impossibility result may be obtained for discrete time dynamics, i.e., where the dynamics are given by a continuous function $f: X\to X$. The proof is similar to that of Theorem~\ref{thm:imposs}, modified to use the appropriate Conley index theory for discrete time dynamics \cite{mrozek1990leray, franks2000shift, richeson1998connection}. In the discrete time case, for an appropriate notion of index pair $(N,L)$ for an isolated invariant set $S$, there is an induced map $f_{N,L}:H_\bullet(N,L)\to H_\bullet(N,L)$, where $H_\bullet(N,L)$ is a singular homology using field coefficients. Define
\[
CH_\bullet(S) := \bigcap_{n>0} f_{N,L}^n\big(H_\bullet(N,L)\big),
\]
and take $\chi: CH_\bullet(S)\to CH_\bullet(S)$ to be the automorphism induced by $f_{N,L}$. In this case, the Conley index of $S$ is denoted $\Con(S)$ and defined as $\Con_\bullet(S):= (CH_\bullet(S),\chi_\bullet(S))$; this pair is independent of index pair (up to isomorphism) and is an invariant of $S$.  For the game  \eqref{eqn:game}, postulating discrete time dynamics $f$ instead of a semi-flow $\varphi$, we would have $CH_\bullet(A)$ and $CH_\bullet(X)$ as in \eqref{CH:A} and \eqref{CH:X} (where $\Z$ is replaced with a copy of the field), and in addition equipped with automorphisms $\chi_\bullet(A), \chi_\bullet(X)$. However, as $CH_\bullet(A)\neq CH_\bullet(X)$,  the result again follows from the principle that if $\Con_\bullet(A)\neq \Con_\bullet(X)$, then $\Con_\bullet(R)\neq 0$, and thus $R\neq \emptyset$.  This principle is proved with an appropriate long exact sequence of $\Con_\bullet(A), \Con_\bullet(X)$ and $\Con_\bullet(R)$, see \cite[Proposition 3.1]{richeson1998connection}, and is analogous to the proof of Theorem~\ref{prop:r0} and Corollary~\ref{cor:neq}.
\end{remark}

For $g$ given by \eqref{eqn:game}, we have shown that $NE(g)$ cannot be the maximal attractor of a dynamical system.  It may, however, be the case that $NE(g)$ is an attractor for some $\varphi$, with non-empty dual repeller.  In this case, we have $CH_\bullet(A)$ as in \eqref{CH:A}, $CH_\bullet(X)$ as in \eqref{CH:X}, and we may compute $CH_\bullet(R)$ with an elementary homological algebra argument. In fact, $CH_\bullet(R)$ is the Conley index of an unstable fixed point (see Fig. \ref{fig:conley_indices}).
\begin{proposition}\label{prop:conley_example}
If $\varphi$ is a (semi)-flow with attractor-repeller pair $(A,R)$ such that $CH_\bullet(A)$ and $CH_\bullet(X)$ are given by \eqref{CH:A} and \eqref{CH:X} respectively, then
\begin{align}\label{CH:R}
 CH_i(R) = \begin{cases} 
      \Z & i = 2 \\
      0 & i \neq 2.
   \end{cases}
\end{align}
\end{proposition}
\begin{proof}
The Conley indices fit together in an exact sequence:
\[
\cdots \xrightarrow{\partial_{n+1}} CH_n(A)\xrightarrow{i_n} CH_n(X)\xrightarrow{j_n} CH_n(R)\xrightarrow{\partial_n} 
CH_{n-1}(A)\xrightarrow{i_{n-1}} \cdots
\xrightarrow{j_0} CH_0(R)\to 0.
\]
By hypothesis for $i\geq 3$ the sequence takes the form
\[
\cdots \to 0\xrightarrow{j_i} CH_i(R)\xrightarrow{\partial_i} 0\to \cdots
\]
By exactness $0 = \img j_i = \ker \partial_i = CH_i(R)$.   For $i = 2$, the sequence has the form:
\[
\cdots \to 0 \xrightarrow{j_2} CH_2(R) \xrightarrow{\partial_2} \Z \xrightarrow{i_1} 0 \to \cdots
\]
Note that $\partial_2$ is injective ($\ker\partial_2 = \img j_2 = 0$) and surjective $( \img \partial_2 = \ker i_1 = \Z$).  Therefore $\partial_2$ is an isomorphism and $CH_2(R)\cong \Z$.   Finally, the remainder of the sequence is 
\[
\cdots \to 0\xrightarrow{j_1} CH_1(R) \xrightarrow{\partial_1} \mathbb{Z}\xrightarrow{i_0} \mathbb{Z}\xrightarrow{j_0} CH_0(R)\to 0.
\]
The map $i_0$ is an isomorphism, mapping the single connected component of $CH_0(A)$ to $CH_0(X)$.  By exactness, $0=\img j_1 = \ker\partial_1$.  Thus $CH_1(R)\cong \img \partial_1 = \ker i_0 = 0$, where the last equality follows as $i_0$ is an isomorphism. Similarly, $\Z = \img i_0 = \ker j_0$.  Thus $\img j_0 = 0$.  However, $j_0$ is surjective, since by exactness $CH_0(R) = \img j_0$.  Therefore $CH_0(R) = 0$, which completes the proof.
\end{proof}


One might imagine weakening Type III to a Type III$_A$, which instead requires that the set of Nash equilibria $NE(g)$ (taken together with the connecting orbits between) forms an attractor $A$; a Type III game is then necessarily Type III$_A$. We can show an incompatibility result for flows in this case (we expect the result also holds for semi-flows).

\begin{theorem}[Incompatibility Theorem]
There exists a game $g$ such that no flow $\varphi$ for $g$ can be both Type III$_A$ and Type II.
\end{theorem}
\begin{proof}
We consider again the game $g$ given in \eqref{eqn:game}.  If $\varphi$ is Type III$_A$, then the same argument as in the proof of Theorem~\ref{thm:imposs} shows that $\varphi$ has an attractor-repeller pair $(A,R)$ such that $CH_\bullet(A)$, and $CH_\bullet(X)$ are given by \eqref{CH:A} and \eqref{CH:X}, and it follows from Proposition~\ref{prop:conley_example} that $CH_\bullet(R)$ is given by \eqref{CH:R}.   It follows from Proposition \ref{prop:fp_conley} that $R$ contains a fixed point.  Thus $NE(g)\subsetneq Fix(\varphi)$, and $\varphi$ cannot be a Type II flow.
\end{proof}

\section{Nondegenerate games and approximate equilibria}

\subsection{Nondegenerate games}

Our impossibility result in the previous section is constructed around a degenerate normal-form game with a continuum of equilibria.  
What if the game is nondegenerate?

\begin{theorem}\label{thm:nondeg_dynamics}
For any nondegenerate game $g$ there is a Type III dynamical system $\varphi_g$.
\end{theorem}
\begin{proof}
Since $g$ is nondegenerate, it has an odd number of isolated Nash equilibria \citep{Shapley1974}. Fix one such equilibrium and call it $y$.  We next define $\varphi_g$ in terms of $y$. We shall define it at point $x$ implicitly, in terms of the direction of motion, and the speed of motion; if this is done, $\varphi_g(t,x)$ is easily computed through integration on $t$.  The direction of motion is the unit vector of $y-x$: the dynamics heads to $y$.  The speed of motion is defined to be $c\cdot D_g(x)$, where $c>0$ is a constant, and by $D_g(x)$ we denote the {\em deficit at $x$:} the sum, over all players, of the difference between the best-response utility at $x$, and the actual utility at $x$.  It is clear that $D_g(x)\geq 0$, and it becomes zero precisely at any Nash equilibrium.

Now it is easy to check that $\varphi_g$ is a well defined dynamics. Furthermore, since the underlying flow is acyclic in a very strong star-like sense, its chain recurrent set coincides with the set of its fixed points, which coincides with the set of Nash equilibria of $g$ --- since there is no opportunity to close extra cycles by $\epsilon$ jumps --- completing the proof.
\end{proof}

Now note that the algorithm for specifying $\varphi_g$ requires finding $y$, a PPAD-complete (and FIXP-complete for more than two players) problem. We believe that the dependence is inherent:

\begin{conjecture}
The computational task of finding, from  the description of a game $g$, either a degeneracy of $g$ or an algorithm producing the direction of motion and speed of a dynamical system of Type III is PPAD-hard (and FIXP-hard for three or more players).
\end{conjecture}

We believe this is an important open question in the boundary between computation, game theory, and the topology of dynamical systems, whose resolution is likely to require the development of new complexity-theoretic techniques pertinent to dynamical systems.

\subsection{$\epsilon$-approximate Nash equilibria}

Next, it may seem plausible that the difficulty of designing dynamics that converge to Nash equilibria can be overcome when only $\epsilon$-approximation is sought, for some $\epsilon>0$.  Recall that an $\epsilon$-Nash equilibrium is a mixed strategy profile in which all players' utilities are within an additive $\epsilon>0$ of their respective best response. We go on to show that our impossibility theorem extends to this case as well. Let us denote by $NE_\epsilon(g)$ the set of $\epsilon$-approximate Nash equilibria of a game $g$. Finally, let us call the dynamics $\varphi$ for a game $g$ to be of Type III$_\epsilon$ if $CR(\varphi)=NE_{\epsilon}(g)$.

\begin{theorem}
There is a game $g$ which admits no Type III$_\epsilon$ dynamics.
In fact, the set of games which admit no Type III$_\epsilon$ dynamics has positive measure.
\end{theorem}

\begin{proof}
Consider again the Kohlberg-Mertens game \eqref{eqn:game}.
We claim that its set of $\epsilon$-Nash equilibria $NE_{\epsilon}(g)$ is homotopy equivalent to $S^1$ (a circle) for
sufficiently small $\epsilon > 0$.  To prove the claim, we subdivide the set of all strategy profiles into nine polytopes $P^{ij}\ $ for all $ i,j \in\{1,2,3\}$, where the polytope $P^{ij}$ is the subset of the strategy space such that the best response of player 1 is $i$, and that of player 2 is $j$. Obviously, these regions can be defined by linear inequalities.  Let $P^{ij}_{\epsilon}$ denote the intersection of $P^{ij}$ and $NE_{\epsilon}(g)$; it can be defined by the linear inequalities defining $P^{ij}$ plus the inequality stating that $(x,y)$ is in $NE_\epsilon(g)$:
\[ x^T M_1 y \geq (M_1 y)_i - \epsilon, \]
and a similar inequality for player 2.
\footnote{As an example, for $P^{22}_{\epsilon}$ and $\epsilon=0.27$, the solution of these inequalities is of the form
\[
x_1=0\land \left(\left(0\leq x_2<0.73\land y_1=0\land \frac{x_2-0.73}{x_2-1}\leq y_2\leq 1\right)\lor \left(0.73\leq x_2\leq 1\land y_1=0\land 0\leq y_2\leq 1\right)\right) .
\]}
Now it is clear that $NE_{\epsilon}(g)$ is the union of these nine manifolds.
However, it is known for the Kohlberg-Mertens game  that strategies 2 and 3 are weakly dominated by the first strategy for the first player (and by symmetry, also for the second player) \citep{DemichelisRitzberger}.\footnote{This can straightforwardly be observed by verifying that the following system of inequalities is tautological:
\[
\begin{cases}
2y_1 + y_2 - 1 \geq y_2 - 1 \\
2y_1 + y_2 - 1 \geq 3y_1 + 2y_2 - 2.
\end{cases}
\]}
Hence, all manifolds $P^{ij}_\epsilon$ are contained within $P^{11}_\epsilon$. Thus $NE_\epsilon(g)$ is a connected, compact $4$-dimensional manifold (with boundary) which is homotopy equivalent to $NE(g)$ for a sufficiently small $\epsilon$.\footnote{Fig. \ref{fig:approx_projection} shows a projection (in a particular direction of $\R^4$) of $NE_\epsilon(g)$ for $\epsilon = 0.09$.  See the supplementary material for a video of {\em 3-dimensional slices} of $NE_\epsilon(g)$. 
Computations performed using Mathematica show that $NE_\epsilon(g)$ is homotopy equivalent to $S^1$ up to at least $\epsilon = 0.09$, and it becomes homotopy equivalent to a ball for some $\epsilon \in (0.09, 0.12)$.}
Assuming $NE_\epsilon(g)$ is an attractor for a dynamical system $\varphi$, we may reason similarly to the proof of Theorem~\ref{thm:imposs} and again invoke Corollary~\ref{cor:neq} to show that there cannot exist Type III$_\epsilon$ dynamics for $g$.

To show that the set of such games which admit no Type III$_\epsilon$ dynamics has positive measure,
consider the approximation problem for an $\epsilon$ substantially smaller than the limit near $0.09$ --- say $\epsilon = 0.03$ --- and all perturbations of the same game where each utility value of the normal form of $g$ in \eqref{eqn:game} is perturbed independently, so that the 18-dimensional vector of perturbations has norm $\frac{\epsilon}{c}$, for some appropriately large constant $c$.  It is clear that this set of games has positive measure.  Let us consider a game $g'$ in this ball. First, it is self-evident that any equilibrium of $g$ is an $\epsilon$-equilibrium of $g'$, hence the set of $NE_{\epsilon}(g')$ contains $NE(g)$.
Furthermore, it is also clear that any strategy profile that is {\em not} a $0.09$-equilibrium of $g$
is not in $NE_{\epsilon}(g')$. 
Thus $NE_\epsilon(g')$ is contained in $NE_{0.09}(g)$.  It follows that $NE_\epsilon(g')$ is homotopy equivalent to $S^1$, and thus the argument above holds for $g'$, which completes the proof.
\end{proof}

\begin{figure}[ht]
    \centering
    \includegraphics[scale=0.4]{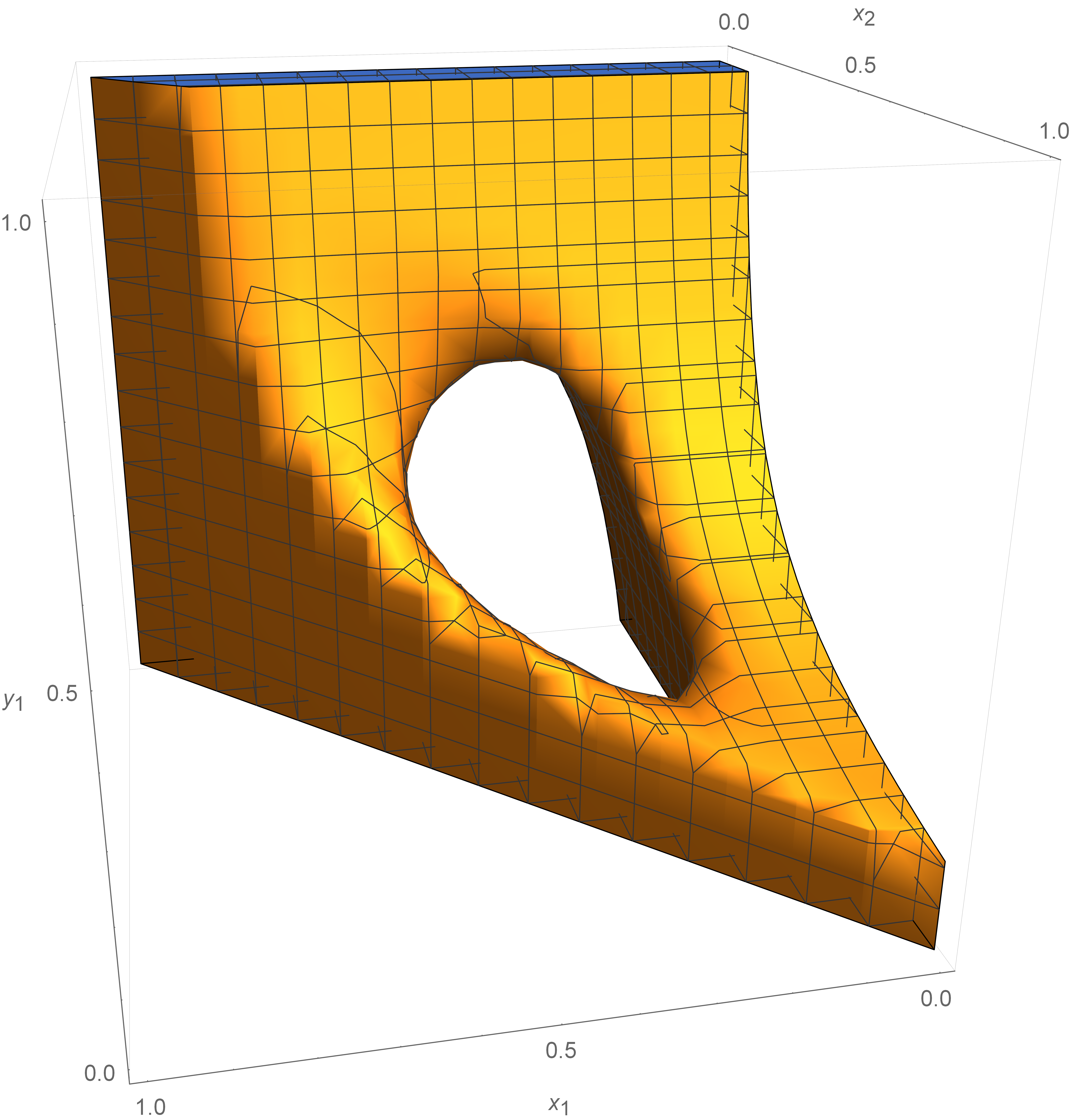}
    \caption{3-dimensional projection of $NE_\epsilon(g)$ when $g$ is the Kohlberg-Mertens game \eqref{eqn:game}, and for utility-normalized $\epsilon=0.09$. The object depicted is homotopy equivalent to $S^1$ (a circle).}
    \label{fig:approx_projection}
\end{figure}

\begin{remark}
A generic perturbation of the Kohlberg-Mertens game has a finite set of isolated Nash equilibria (and is nondegenerate), and we know from Theorem~\ref{thm:nondeg_dynamics} that Type III dynamics do exist for this perturbed game. It may, therefore, appear surprising that we can prove a stronger impossibility result (positive measure of such games) despite the goal being more modest (just approximation instead of an exact equilibrium). The reason is that as soon as the sought approximation $\epsilon$ becomes much larger than the perturbation of the game (equivalently, the perturbation of the game being much smaller than the approximation), and it is required that all approximate equilibria be the only chain recurrent points of the dynamics, $NE_\epsilon(g)$ is once again is homotopy equivalent to $S^1$, and our characterization of the Conley indices (Corollary~\ref{cor:neq}) once again applies.
\end{remark}

\section{Conclusion}

In this paper we have argued that the notion of Nash equilibria is fundamentally incomplete for describing the global dynamics of games.  More precisely, we have shown that there are games where the set of Nash equilibria does not comprise the entire chain recurrent set, and thus the Nash equilibria cannot account for all of the long-term dynamical behavior.  Moreover, this is true even when one relaxes the focus from Nash equilibria to approximate Nash equilibria.

We have utilized the chain recurrent set in this paper in order to characterize game dynamics. However, ultimately we believe that it is not the right tool for the analysis of game dynamics.  The chain recurrent set is a brittle description of the fine structure of dynamics, i.e. not robust under perturbations.  In contrast, the rules which are suggested to govern players' behavior are not meant as precise models derived from first principles, but instead as rough approximations.  Thus the appropriate mathematical objects for analysis for game dynamics must be robust to perturbation; in the words of~\citep{conley1978isolated}, \emph{"...if such rough equations are to be of use it is necessary to study them in rough terms"}. 
Instead, we propose a focus on the coarser concept of \emph{Morse decomposition} (if a system has a finest Morse decomposition, then it is the chain recurrent set).  Morse decompositions are posets of isolated invariant sets which possess a duality theory with lattices of attractors, and in addition have an associated homological theory using the Conley index \citep{conley1978isolated}. Ultimately, these ideas culminate in theory of the connection matrix, which describes the global dynamics by means of a chain complex, and which would provide a robust, homological theory of the global dynamics of games.

The intersection of these ideas with the solution concepts of online learning and optimization (e.g., regret) as well as that of game theory (e.g., Nash/(coarse) correlated equilibria) holds the promise of a significantly more precise understanding of learning in games. 
Although these future theories are yet to be devised, one of their aspects is certain: Nash equilibria are not enough.

\section*{Acknowledgements} 

The work of K.S. was partially supported by EPSRC grant EP/R018472/1. K.S. would like to thank Rob Vandervorst for numerous enlightening discussions regarding Conley theory.

\bibliography{refs,refs2}

\appendix

\section{A primer on algebraic topology}

\subsection{Introductory remarks and definitions}

The prime idea of algebraic topology is to use algebraic invariants -- which may not just be numbers, but can also be sets, groups, vector spaces, and so on -- to structurally distinguish between topological spaces, i.e., so that we may conclude when ``two spaces are different.'' The idea of an invariant is that even though we may apply some reasonable transformation of a space to another (see below, e.g., homeomorphism, or homotopy), the transformed space is expected to retain certain properties/groups of the original one. 

First, a {\em topological space} is a set of points along with a ``structure'' on it, called topology on the set, which is formally defined by the acceptable neighborhoods of any point. The neighborhoods of a point allow us to define the closeness of points, and furthermore, continuity and limits. We will not elaborate on the definitions of those (which belong to the branch of point set topology), since in our case, we will only have to deal with subsets of metric spaces, where we have a well-defined norm in the usual sense.

Recall the definition of a {\em quotient set}, i.e., the set $S/\sim$ of all the equivalences classes of a binary relation $\sim$ under $S$. Commonly, for a subset $A\subseteq S$, $S/A$ is the quotient set defined by the natural relation that identifies all points within $A$ as equivalent to one another, and no other point in $S\setminus A$ as equivalent to any other point. Then, the {\em quotient space} $X/\sim$ (sometimes also referred to as $X$ modulo the equivalence relation $\sim$) for a topological space $X$ is defined as the corresponding quotient set, endowed with the finest topology that unifies the neighborhoods of equivalent points. For example, if we take the line $[0,1]$ {\em modulo} the relation that identifies its boundary (which are the two points, 0 and 1), then the resulting space is ``equivalent'' to a unit circle, denoted as $S^1$. What is our notion of ``equivalence'' here? It is the standard one, called {\em homeomorphism} in topology.

\begin{definition}
\label{def:homeo}
Two topological spaces $X, Y$ are {\em homeomorphic} (denoted by $X \cong Y$) if there exist continuous maps $f:X\to Y, g:Y\to X$ such that $f\circ g = id_Y$ and $g\circ f = id_X$, where $id_X: X\to X, id_Y: Y\to Y$ are identity maps.
\end{definition}

Often the notion of homeomorphism is too strong. This is why a new, relaxed notion based upon the {\em continuous deformation} of a space to another is necessary, called {\em homotopy equivalence} of spaces. Before we give this definition, we start with the definition of the notion of deforming a function into another function:

\begin{definition}
Let $X, Y$ be two topological spaces. We say that two continuous maps $f, g: X \to Y$ are {\em homotopic} (denoted by $f\sim g$) if there exists a map $H: X\times [0,1] \to Y$ such that $H(x, 0) = f(x)$ for all $x\in X$ and $H(x, 1) = g(x)$ for all $x\in X$.
\end{definition}

One important observation in this definition is that it is not allowed for the continuous map $H$ to ``exit'' $Y$, but it must constantly remain within it, i.e., $H(x,t)\in Y$ for all $x\in X, t\in [0,1]$.

\begin{definition}
Two topological spaces $X, Y$ are {\em homotopy equivalent} (denoted by $X \simeq Y$) if there exist continuous maps $f:X\to Y, g:Y\to X$ such that $f\circ g \sim id_Y$ and $g\circ f \sim id_X$, where $id_X: X\to X, id_Y: Y\to Y$ are identity maps.
\end{definition}

Note that the difference with the homeomorphism is that here the functions are only required to be able to be ``continuously deformed'' into the identities, and not necessarily be exactly equal to them.

Finally, we review some notions from standard group theory. First, the definition of when two groups are {\em isomorphic} closely follows \Cref{def:homeo}. A {\em group homomorphism} is a function $h: G \to H$ for two groups $(G, \star)$ and $(H, \cdot)$ that preserves the group structure, i.e., $h(u\star v) = h(u) \cdot h(v)$ for all $u,v\in G$. The {\em kernel} $\ker h$ of a group homomorphism $h$ is the set $\{ u\in G : h(u) = e_H\}$ where $e_H$ is the identity element of $H$. The {\em image} $\img h$ of a group homomorphism $h$ is the set $\{ h(u) : u\in G\}$.  Whenever we refer to abelian groups, i.e., groups with the commutative property $a + b = b + a$, we usually utilize the addition notation, i.e., the group operation is $+$, the identity element is $0$, and the inverse of an element $x$ is denoted by $-x$. In this case, we also denote by 0 the {\em trivial group}, i.e., the group that consists of only the identity element.

The {\em quotient group} $G/H$ when $H$ is a subgroup of an abelian group $G$ is defined as the set of all cosets of $H$ in $G$, along with the natural binary operation of the group extended into cosets. Finally, the {\em free abelian group with a certain basis} is the set of all \textit{integer} linear combinations of elements of that basis, alongside a natural operation that combines these. For example, if the basis is $\{\sigma_1,\sigma_2,\dots,\sigma_n\}$, then each element of the free abelian group with that basis can be written in the form $\sum\limits_{i=1}^n n_i \sigma_i$ for some $n_i\in\Z$, where the operation of the group is addition ($+$).

\subsection{Singular homology}
\label{app:singular}

We move on to briefly review singular homology. A standard reference is given by \citet{hatcher}. Intuitively, the reason to develop the theory through the dependence of a topological space on simplices is that for simplices we have an understanding of how to define notions such as ``cycles'' and ``holes.''

Denote by $\Delta^n$ the standard $n$-simplex, spanned by the unit vertices $v_0, v_1, \ldots, v_n$. The boundary of the $n$-simplex consists of $(n-1)$-simplices that are its faces. The {\em $i$th face} is the $n-1$ simplex spanned by the vertices $v_0,\ldots,v_{i-1},v_{i+1},\ldots,v_n$. We also consider the faces of simplices (and, respectively, the faces of those, and so on) as {\em oriented} by the order given by the higher-dimensional simplex from which they have arisen. The way, then, to obtain the correct orientations for enclosing a simplex (e.g., clockwise) is to take the faces in its boundary with alternating signs.

Let $X$ be a topological space. Consider all possible continuous maps $\sigma: \Delta^n\to X$ from the $n$-simplex to our topological space. Any such map is called a \emph{singular $n$-simplex}. The free abelian group with basis this set of all possible singular $n$-simplices is denoted by $C_n(X)$ and called the {\em $n$-th chain group}.
Its elements are called {\em $n$-chains}, i.e., any $n$-chain can be written as $\sum\limits_i n_i \sigma_i$ for some $n_i\in\Z$, where $\sigma_i: \Delta^n\to X$.
The \emph{boundary} of a singular simplex is denoted by $\partial_n(\sigma)$ and defined as 
\[
\partial_n(\sigma) = \sum_{i=0}^n (-1)^i \sigma|_{v_0,\ldots, v_{i-1},v_{i+1}, \ldots, v_n},
\]
where $\sigma|_{v_0,\ldots, v_{i-1},v_{i+1}, \ldots, v_n}$ indicates $\sigma$ restricted to the $i$th face of $\Delta^n$. That is, $\partial_n(\sigma)$ is the formal sum of $(n-1)$ simplices which are the restriction of $\sigma$ to the faces of $\Delta^n$, with an alternating sign to account for orientation. 

The boundary map $\partial_n: C_n(X)\to C_{n-1}(X)$
defined as the extension of boundaries to act on $n$-chains
is a group homomorphism
(since each element of the free abelian group $C_n(X)$ has a canonical representation, it is sufficient to define the map on the basis, as per the definition above, and then extend simply taking the linear combination of the operator results on the basis elements).
It can be shown that $\partial_n\partial_{n+1}=0$, i.e., $\img \partial_{n+1}\subseteq \ker \partial_n$.
The boundary operators $\partial_n$ together with the chain groups $C_n(X)$ form a {\em chain complex} of abelian groups (which simply means that in the representation below, for every two consecutive maps, it holds that $\img \partial_{n+1}\subseteq \ker \partial_n$), called the \emph{singular chain complex}, denoted $(C_\bullet(X),\partial_\bullet)$, or more simply $C_\bullet(X)$:
\[
\ldots \xrightarrow{\partial_{n+1}} C_n(X) 
\xrightarrow{\partial_n} C_{n-1}(X) 
\xrightarrow{\partial_{n-1}} \ldots
\xrightarrow{\partial_2} C_1(X)
\xrightarrow{\partial_1} C_0(X) \xrightarrow{\partial_0} 0.
\]

Intuitively, the $n$-chains that comprise the group $\ker\partial_n$ correspond to $n$-dimensional ``cycles'' in the $X$. Homology is a measurement of the ``holes'' of a space, obtained via ``cycles mod boundary.'' That is, any cycle is that is also a boundary is deemed trivial.  
Formally, the $n$-th \emph{singular homology group} is defined as $H_n(X) = \ker\partial_n / \img \partial_{n+1}$.  The homology of $X$, denoted $H_\bullet(X)$, is the collection of singular homology groups. The reduced homology of $X$, denoted $\tilde{H}_\bullet(X)$, is calculated via the (augmented) chain complex
\[
\ldots \xrightarrow{\partial_{n+1}} C_n(X) 
\xrightarrow{\partial_n} C_{n-1}(X) 
\xrightarrow{\partial_{n-1}} \ldots
\xrightarrow{\partial_2} C_1(X)
\xrightarrow{\partial_1} C_0(X)
\xrightarrow{\epsilon}  \Z \to 0,
\]
where $\epsilon(\sum n_i \sigma_i) = \sum_i n_i$.  Reduced homology often simplifies homology computations.  Indeed, if $X$ is a one point space, then $\tilde{H}_i(X) = 0$ for all $i$.

For a subspace $A\subseteq X$, the relative homology $H_n(X,A)$ is defined as the homology of the quotient chain groups, $H_n(X,A) = H_n(C_n(X)/C_n(A))$.  Intuitively, the relative homology measures the topology of $X$ relative to that of $A$ (i.e., if we deem all ``cycles'' within $A$ as trivial) and if $A$ fulfills the relatively mild condition that there is a neighborhood in $X$ of $A$ that deformation retracts to $A$, then $H(X,A)\cong \tilde{H}(X/A)$, the reduced homology of the quotient space $X/A$~\citep{hatcher}.

A standard tool in algebraic topology is the {\em exact sequence}, which is a sequence of abelian groups  and group homomorphisms:
\[
\cdots \to A_{n+1}\xrightarrow{\alpha_{n+1}}
A_n \xrightarrow{\alpha_n}
A_{n-1} \to \cdots
\]
where $\img \alpha_{n+1} = \ker\alpha_n$ for all $n$.
A fundamental theorem in algebraic topology is that, for a subspace $A\subseteq X$, there is an exact sequence, called the \emph{long exact sequence in homology}, which links $H_\bullet(X,A)$ together with $H_\bullet(X)$ and $H_\bullet(A)$ \citep{hatcher}:
\[
\cdots \xrightarrow{\partial_{n+1}} H_n(A)\xrightarrow{i_n} H_n(X)\xrightarrow{j_n} H_n(X,A)\xrightarrow{\partial_n} 
H_{n-1}(A)\xrightarrow{i_{n-1}} H_{n-1}(X)\xrightarrow{j_{n-1}} \cdots 
\xrightarrow{j_0} H_0(X,A)\to 0.
\]

\end{document}